\documentclass[letterpaper, 10 pt, conference]{ieeeconf}  

\IEEEoverridecommandlockouts                              
\usepackage[utf8]{inputenc}
\usepackage{cite}

\let\labelindent\relax
\usepackage{amsmath, amssymb, amsthm, amsfonts}
\usepackage{graphicx}
\usepackage{algorithm}
\usepackage{algpseudocode}
\usepackage{xcolor}
\usepackage{booktabs}
\usepackage{enumitem}
\usepackage[hidelinks]{hyperref}
\usepackage{placeins} 
\theoremstyle{plain}
\newtheorem{theorem}{Theorem}
\newtheorem{proposition}{Proposition}
\newtheorem{lemma}{Lemma}

\newtheorem{assumption}{Assumption}
\theoremstyle{definition}
\newtheorem{definition}{Definition}

\theoremstyle{remark}

\DeclareMathOperator*{\Proj}{Proj}
\newcommand{\SafeFilter}{\Gamma}           

\newcommand{\Psafe}{P_{\mathrm{safe}}}

\title{\LARGE \bf
A Stackelberg Game Framework with Drainability Guardrails for Pricing and Scaling in Multi-Tenant GPU Cloud Platforms
}

\author{Junji Yan, Asrin Efe Yorulmaz, Hanchen Zhou, and Tamer Ba\c{s}ar 
\thanks{This work was supported in part by the IBM-ILLINOIS Discovery Accelerator Institute (IIDAI) under Grant \#122955 and in part by the AFOSR Grant FA9550-24-1-0152.  Junji Yan, Asrin Efe Yorulmaz, and Tamer Ba\c{s}ar are with the Coordinated Science Laboratory (CSL), University of Illinois Urbana-Champaign (UIUC), Urbana, IL 61801, USA. Hanchen Zhou was also with CSL, UIUC; he is now with CFINS, Department of Automation, Tsinghua University, Beijing 100084, China. Emails: \{junjiy2, ay20, basar1\}@illinois.edu, zhouhc22@mails.tsinghua.edu.cn.}
}

\begin{document}

\maketitle


\thispagestyle{empty}
\pagestyle{empty}

\begin{abstract}
Modern Graphics Processing Unit (GPU)-backed services must satisfy strict latency service-level objectives (SLOs) while controlling spare-capacity costs. In multi-tenant GPU cloud platforms, this trade-off is inherently dynamic because workload demand is endogenous; specifically, pricing shapes the submissions of heterogeneous tenants, which subsequently impact congestion and delay. We formulate the joint pricing-and-scaling problem as a large-population Stackelberg game problem, and we derive an explicit equilibrium demand map. The resulting closed-loop model reveals a structural failure mode in which delay-insensitive workloads sustain a residual demand floor, making the backlog undrainable under bounded price and service capacity. This observation motivates a computable drainability guardrail that certifies uniformly negative backlog drift in the residual-demand regime. For any fixed price-capacity pair satisfying the drainability guardrail, we establish global convergence to a unique operating point under a checkable step-size condition. Building on this fixed-pair analysis, we further develop an optimizer-agnostic action shield that provides a negative-drift certificate for shielded execution in the residual regime of the dynamic problem and show empirically that it improves safety and robustness for model-free reinforcement learning (RL) in this setting.
\end{abstract}

\section{Introduction}
\label{sec:intro}

In recent years, latency-oriented service-level objectives (SLOs) have become a critical requirement for many Graphics Processing Unit (GPU)-backed services.
Meeting these SLOs under bursty demand often requires keeping spare GPU capacity online, which can substantially increase costs \cite{gandhi2012autoscale}. GPU platforms therefore face a persistent trade-off between service quality and the cost of excess capacity, particularly when capacity adjustments are subject to provisioning frictions.

Most existing SLO-aware serving and autoscaling systems focus on the supply side and treat workload demand as externally given \cite{zhang2020enabling,patke2024queue}.
In multi-tenant GPU cloud platforms, however, demand is often endogenous and shaped jointly by posted prices and congestion.
The posted price acts as a direct demand-side lever, influencing which tenants participate and how aggressively they submit work, while congestion-induced delay feeds back into demand through service quality. Consequently, SLO-aware operation requires managing capacity and demand jointly in a closed loop.

To address this gap, we study a joint price--capacity control problem for a multi-tenant GPU cloud platform in a closed loop, aiming to maximize long-run profit while accounting for latency SLO violations.
We adopt a Stackelberg game formulation \cite{bacsar1998dynamic} to make the demand-side feedback explicit. At each decision epoch, the platform, leader, posts a unit price and a capacity target, and tenants, followers, respond by choosing submission rates based on the posted price and the public congestion signal.
To capture heterogeneity, we approximate the population by finitely many mean-field types \cite{BasarDjehicheTembine2026} with distinct price and delay sensitivities, including delay-insensitive but price-sensitive types that represent opportunistic or speculative demand.

This Stackelberg formulation also reveals a structural failure mode that is absent in exogenous-demand formulations and therefore prevents a direct move to nominal dynamic optimization. That is, under admissible leader actions, the system may become undrainable in the high-congestion regime. 
First, although demand can change rapidly, physical GPU capacity adjusts slowly due to provisioning frictions, such as cold starts \cite{san2023reducing}, setup delays \cite{Gandhi-QUESTA14}, and scale-down inertia \cite{lorido2014review}. 
Second, tenant demand responses to price and congestion are heterogeneous and need not vanish uniformly as congestion grows. Specifically, while delay-sensitive users eventually drop out, delay-insensitive types continue submitting work, sustaining a residual demand floor.
These observations motivate characterizing the admissible leader actions under which the backlog remains in a stable and drainable regime.

Such a structural characterization is a prerequisite for dynamic learning.
Standard model-free reinforcement learning (RL) relies on trial-and-error exploration \cite{sutton1998reinforcement} and may sample actions that push the system into an undrainable regime.
Explicitly identifying the drainable action region thus provides a theoretical guardrail that naturally translates into an action shield \cite{alshiekh2018safe}, eenabling safer exploration of joint pricing and scaling policies while preserving drainability, pursuing long-run profit, and mitigating SLO violations.

\addvspace{0.5\baselineskip}
\noindent\textbf{Related Work.}\hspace{0.5em}
Our work relates to several lines of research on cloud control, pricing, and congestion-sensitive demand. 
SLO-aware serving and autoscaling systems study capacity control with exogenous
demand under provisioning frictions~\cite{zhang2020enabling, patke2024queue}.
GPU pricing systems focus on pricing and charging models for GPU resources \cite{Agora-arxiv25}.
Dynamic pricing under congestion has also been studied in queueing systems~\cite{masuda1999dynamic,ccelik2008dynamic,katta2005pricing}. Related leader--follower formulations appear in large-population network settings \cite{BS-JOTA02}, in congestion pricing with heterogeneous users \cite{stidham2004pricing}, and in cloud pricing and provisioning models \cite{DiValerio,gera2011learning}.
Building on these ideas, we study joint pricing and scaling in SLO-oriented multi-tenant GPU cloud platforms with endogenous demand and provisioning frictions.
\addvspace{0.5\baselineskip}
\noindent\textbf{Contributions.}\hspace{0.5em}
We make the following contributions:
\begin{itemize}[leftmargin=2em]
  \item 
  We formulate joint pricing and scaling in multi-tenant GPU cloud platforms with provisioning frictions as a large-population Stackelberg game problem, yielding an endogenous equilibrium demand map that responds to pricing and congestion.

  \item
   For the reduced fixed price-capacity recursion, we identify an undrainability mechanism in the high-congestion regime and derive a computable drainability guardrail to certify negative backlog drift. Under this guardrail we prove that the fixed-pair recursion admits a unique operating point with global convergence under a checkable step-size condition. 

  \item 
    We translate this drainability guardrail into an optimizer-agnostic action shield for the dynamic problem, establish a  negative-drift certificate for shielded execution, and demonstrate improved safety and robustness for model-free RL.
\end{itemize}

\section{Model and Problem Formulation}
\label{sec:model}

We consider discrete decision epochs $t=0,1,\dots$ with interval $\Delta>0$.
At each epoch, the platform, leader, first observes the current system state $s_t$ and chooses an action
$a_t=(P_t,S_t^{\mathrm{tar}})$ consisting of a unit price and a capacity target.
Meanwhile, the chosen target capacity requests enter a pending setup pipeline and gradually translate into active service capacity, while scale-down remains inertial. 
Given the posted price $P_t$ and a public congestion signal $D_t$, tenants, followers, choose stage-wise submission rates, which generate endogenous demand.
The resulting aggregate demand enters the queueing recursion, interacts with the currently active capacity, and determines the next-state update.

The state is defined as
\(
s_t := (Q_t,S_t,B_t,S_t^{\mathrm{tar,prev}})\in\mathcal S
\),
where $Q_t$ is the backlog, $S_t$ is the active GPU capacity, $B_t$ is the pending setup pipeline, and $S_t^{\mathrm{tar,prev}}$ records the previous capacity target to capture target-change dynamics. Here, the setup pipeline represents requested capacity that has not yet become active.
It decreases through activation and may be partially canceled after target decreases.
The feasible state space is
\(\mathcal S := \mathbb{R}_{+} \times [0,S_{\max}] \times [0,B_{\max}] \times [0,S_{\max}].\)
The action of the leader is defined as \(a_t := (P_t,S_t^{\mathrm{tar}})\in\mathcal A\),
where $P_t$ is the posted unit price and $S_t^{\mathrm{tar}}$ is the capacity target. The feasible action space is
\(\mathcal A := [P_{\min},P_{\max}] \times [0,S_{\max}],  P_{\min}>0\).
We model the public congestion signal using the queue-based delay proxy \cite{kim2018value}:
\begin{equation}
D(Q_t,S_t) := \frac{Q_t}{S_t\mu+\epsilon},
\label{eq:delay_proxy}
\end{equation}
where $\mu>0$ is the service rate per GPU and $\epsilon>0$ is a small baseline constant. For brevity, we denote the delay proxy $D(Q_t,S_t)$ as $D_t$.  

\subsection{Large-Population Demand Response}
\label{sec:model_demand}

We model a large population of tenants partitioned into $K$ mean-field types with density
$\rho_k\ge 0$ and $\sum_{k=1}^K\rho_k=1$.
At each epoch, given the posted price $P_t$ and the public congestion signal $D_t$, we assume that each type chooses its submission rate myopically. The monotone stage-wise response enables the subsequent uniqueness and convergence analysis; delayed or non-monotone tenant adaptation falls outside these guarantees. We further assume that the coupling across types arises only through the shared congestion signal and the resulting queue recursion. Type $k$ is characterized by
\begin{equation*}
\theta_k := (w_k,\mathrm{SLO}_k),\ \ 
w_k>0,\ \ \mathrm{SLO}_k\in(0,\infty],
\end{equation*}
where smaller $\mathrm{SLO}_k$ implies higher delay sensitivity, and $\mathrm{SLO}_k=\infty$ represents delay-insensitive residual types.
We encode delay sensitivity via the SLO-risk coefficient:
\(\delta_k\ :=\ \alpha/\mathrm{SLO}_k, \alpha>0\),
and thus $\mathrm{SLO}_k=\infty$ yields $\delta_k=0$. Given $(P_t,D_t)$, a representative type-$k$ tenant selects a submission rate $\lambda\ge 0$ by solving
\begin{equation*}
\max_{\lambda\ge 0}\ \ 
 w_k\log(1+\lambda) - (P_t+\delta_k D_t)\lambda,
\end{equation*}
where $w_k$ scales the utility of type $k$. The concave $\log(1+\lambda)$ utility captures diminishing returns, and the linear term captures per-unit monetary and delay-related costs \cite{BasarSrikant02}. This strictly concave problem admits the unique solution
\begin{equation}
\lambda_k^{\ast}(P_t,D_t)
=\Big[\frac{w_k}{P_t+\delta_k D_t}-1\Big]_+ ,
\label{eq:best_response}
\end{equation}
where $[x]_+ := \max\{x,0\}$.
Hence $\lambda_k^\ast$ is non-increasing in both $P_t$ and $D_t$, and for $\delta_k=0$ it is independent of congestion. Aggregating across types yields the endogenous equilibrium demand map \(\Lambda^{\ast}(s_t,a_t)
:= \sum_{k=1}^K \rho_k\,
\lambda_k^{\ast}\!\big(P_t, D(Q_t,S_t)\big).\)

\subsection{Queueing and Frictional Provisioning Dynamics}
\label{sec:model_dynamics}

The backlog follows a discretized \emph{fluid recursion} with endogenous demand $\Lambda^\ast(s_t,a_t)$ and service rate $S_t\mu$:
\begin{equation}
Q_{t+1}
=\Big[\,Q_t+\Delta\big(\Lambda^{\ast}(s_t,a_t)-S_t\mu\big)\Big]_+ .
\label{eq:queue_dyn}
\end{equation}
We model capacity provisioning through a pending pipeline: new capacity requests enter the pipeline due to cold starts and setup delays and are gradually activated, while scale-down of active capacity remains inertial.
Let $\omega,\nu\in(0,1]$ and $\chi\in[0,1]$, where $\omega$ is the activation rate of the pipeline, $\nu$ captures scale-down inertia \cite{herlich2016delayed}, and $\chi$ is the fraction of a target-decrease that can cancel pending pipeline.
We decompose target changes into the target-increase component $U_t$ and the target-decrease component $V_t$ as 
\(
U_t := [S^{\mathrm{tar}}_t-S^{\mathrm{tar,prev}}_t]_+, 
V_t := [S^{\mathrm{tar,prev}}_t-S^{\mathrm{tar}}_t]_+.
\)
Here $U_t$ adds requests to the pending pipeline, while $V_t$ triggers cancellation of a portion of pending pipeline if there is any.
The activated capacity in one epoch is
\begin{equation}
A_t := \min\{\omega B_t,\; S_{\max}-S_t\},
\label{eq:At_def}
\end{equation}
and the canceled pending pipeline mass is
\begin{equation}
C_t :=
\min\{\chi V_t,\; B_t-A_t+U_t\}.
\label{eq:C_def}
\end{equation}
Then, pending pipeline and active capacities evolve as
\begin{align}
B_{t+1} &=
\Proj_{[0,B_{\max}]}\!\big(B_t-A_t+U_t-C_t\big),
\label{eq:B_dyn}\\
S_{t+1} &=
\Proj_{[0,S_{\max}]}\!\big(S_t+A_t-\nu [S_t-S^{\mathrm{tar}}_t]_+\big),
\label{eq:S_dyn}
\end{align}
where $\Proj_{[a,b]}(x):= \min\{\max\{x,a\},b\}$ denotes projection onto $[a,b]$. The deterministic fluid model yields nominal drift guarantees; stochastic arrival or service fluctuations would instead require probabilistic or robust safety certificates.
\subsection{The Leader's Control Problem}
\label{sec:model_control}

Given the stage-wise demand response and the induced deterministic state transition,
the leader seeks a policy $\pi=\{\pi_t:\mathcal{S}\to\mathcal{A}\}_{t\ge 0}$ to maximize the infinite-horizon discounted return
\begin{equation*}
J(\pi) := \sum_{t=0}^\infty \gamma^t\, r_t(s_t,a_t),
\ \  a_t = \pi_t(s_t),\ \ \gamma\in(0,1),
\end{equation*}
where $s_{t+1}=f(s_t,a_t)$ is induced by equations \eqref{eq:queue_dyn}--\eqref{eq:S_dyn}.
The one-step reward is defined as follows, with $\Delta$ absorbed into the coefficients:
\begin{align*}
r(s,a)
&:= P\Lambda\!^{\ast}\!(s,a) - c_{\mathrm{op}}S - c_B B
   - \eta_{\mathrm{tar}}\big(S^{\mathrm{tar}}\!\!\!-\!S^{\mathrm{tar,prev}}\big)^2
\notag\\
&\quad - \Phi_0\,\Xi(s,a),
\end{align*}
where $\Phi_0$ is a parameter that scales the SLO-risk penalty, and the SLO-risk component is
\begin{equation*}
\Xi(s,a):=
\sum_{k=1}^K \rho_k\,\lambda_k^{\ast}\!\big(P, D(Q,S)\big)\,\big[D(Q,S)-\mathrm{SLO}_k\big]_+ .
\end{equation*}
The reward balances revenue against operating cost, target-adjustment friction, and SLO-risk penalties. Here $P\Lambda^\ast$ is the revenue term, $c_{\mathrm{op}}S$ charges active GPUs, $c_B B$ charges pending pipeline,
$\eta_{\mathrm{tar}}(\cdot)^2$ discourages frequent target switching, and $\Xi(s,a)$ penalizes excess-delay mass beyond type-specific SLOs.

\section{Fixed-pair backlog recursion and residual regime}
\label{sec:fixed_leader_dynamics}
This section studies the fixed-pair backlog recursion that underlies the guardrail analysis.
Whereas Section~\ref{sec:model} considers the leader dynamically chooses a price--target pair $(P_t,S_t^{\mathrm{tar}})$ under provisioning frictions, here we analyze a reduced fixed-pair model by assuming constant price and service capacity $(P,S)$. 
This reduction lets us focus on the key question of drainability: whether backlog can be driven down from high-backlog states under admissible price and service capacity once delay-sensitive demand drops out.
In this reduced model, $S$ denotes effective service capacity in the induced backlog recursion rather than the target variable of the complete model. 

We first characterize the residual-demand regime, then derive an explicit drainability guardrail, and finally establish operating-point existence, uniqueness, and global convergence under a checkable step-size condition. From the Stackelberg viewpoint, this analysis isolates how a fixed leader price--capacity choice shapes follower-side residual demand, and it will later serve as the local safety certificate behind the online action shield. 

In the fixed-pair model the corresponding recursion of \eqref{eq:queue_dyn} is given by:
\begin{equation}
Q_{t+1}=F_{P,S}(Q_t)
:= 
\Big[\,Q_t+\Delta\big(\Lambda^{\ast}_{P,S}(Q_t)-S\mu\big)\Big]_+,
\label{eq:fixed_backlog_recursion}
\end{equation}
where the stage-wise equilibrium demand is given by
\begin{align*}
    &D(Q):= \frac{Q}{S\mu+\epsilon},\qquad
\lambda_k^{\ast}(Q):= \lambda_k^{\ast}(P,D(Q)),\notag\\
& \quad \qquad \qquad \Lambda^{\ast}_{P,S}(Q):= \sum_{k=1}^K \rho_k \lambda_k^{\ast}(Q).
\end{align*}
Delay-insensitive types with $\delta_k=0$, equivalently $\mathrm{SLO}_k=\infty$, create a residual floor that persists under extreme congestion:
\begin{equation*}
\Lambda_0(P)
:=
\sum_{k:\,\delta_k=0}\rho_k\Big[\frac{w_k}{P}-1\Big]_+ .
\end{equation*}
For delay-sensitive types ($\delta_k>0$), we define the drop-out thresholds as:
\begin{align*}
&Q^{\mathrm{drop}}_k(P,S)
:=
\Big[\,(S\mu+\epsilon)\frac{(w_k-P)}{\delta_k}\Big]_+,\\
&Q^{\mathrm{drop}}(P,S):= \max_{k:\,\delta_k>0}Q^{\mathrm{drop}}_k(P,S).
\label{eq:qdrop_def}
\end{align*}
Then, the following result formalizes the finite-threshold residual regime.
\begin{proposition}
\label{prop:residual_regime}
For any fixed $(P,S)$ and all $Q\ge Q^{\mathrm{drop}}(P,S)$, every delay-sensitive type becomes inactive and
\( 
\Lambda^{\ast}_{P,S}(Q)=\Lambda_0(P).
\) 
\end{proposition}
\begin{proof}
If $Q\ge Q^{\mathrm{drop}}_k(P,S)$ and $\delta_k>0$, then $P+\delta_k D(Q)\ge w_k$, hence $\lambda_k^{\ast}(Q)=0$ by \eqref{eq:best_response}; taking the maximum over $\delta_k>0$ yields the result.
\end{proof}

\subsection{Guardrail and Structural Drainability}
Under fixed $(P,S)$, structural drainability means uniformly negative high-backlog drift.
By Proposition~\ref{prop:residual_regime}, this is equivalent to requiring service capacity to dominate the residual floor. Then, we define the drainability guardrail formally as follows. 

\begin{definition}[Drainability guardrail]
\label{def:guardrail}
A fixed pair $(P,S)$ satisfies the drainability guardrail if
\begin{equation}
\Lambda_0(P) < S\mu.
\label{eq:guardrail}
\end{equation}
Moreover, define residual slack as: \(\varepsilon(P,S):= S\mu-\Lambda_0(P).\)
\end{definition}

The next proposition shows that the guardrail yields uniform negative drift once the system enters the residual regime.
\begin{proposition}
\label{prop:neg_drift}
If \eqref{eq:guardrail} holds, then for all $Q\ge Q^{\mathrm{drop}}(P,S)$,
\(\Lambda^{\ast}_{P,S}(Q)-S\mu
= -\,\varepsilon(P,S)
< 0.\)
\end{proposition}
\begin{proof}
Immediate from Proposition~\ref{prop:residual_regime} and \eqref{eq:guardrail}.
\end{proof}

The next statement characterizes structural non-drainability when the guardrail is violated.
\begin{proposition}
\label{prop:non_drainable}
If $\Lambda_0(P)>S\mu$, then for any $Q_0\ge 0$,
\(Q_t \ge Q_0 + t\Delta\big(\Lambda_0(P)-S\mu\big), \forall t\ge 0,\)
so the backlog diverges, and no operating point exists. On the other hand, if $\Lambda_0(P)=S\mu$, then the residual-regime drift is zero and $Q_{t+1}=Q_t$ for all $t$ once $Q_t\ge Q^{\mathrm{drop}}(P,S)$.
\end{proposition}
\begin{proof}
If $\Lambda_0(P)>S\mu$, then $\Lambda_{P,S}^\ast(Q)\ge \Lambda_0(P)>S\mu$ for all $Q\ge0$, so $F_{P,S}(Q)>Q$ on $\mathbb{R}_+$. Hence no feasible operating point exists. 
If $\Lambda_0(P)=S\mu$, then for any $Q_t\ge Q^{\mathrm{drop}}(P,S)$, Proposition~\ref{prop:residual_regime} gives $\Lambda^\ast_{P,S}(Q_t)=\Lambda_0(P)$. Substituting this into \eqref{eq:fixed_backlog_recursion} yields $Q_{t+1}=Q_t$, so the drift is zero.
\end{proof}

\subsection{Operating Point and Global Convergence}
We formally define an operating point as a fixed point of the recursion $F_{P,S}$, that is, any $Q^\ast \ge 0$ satisfying
\begin{equation*}
Q^{\ast}=F_{P,S}(Q^{\ast}).
\end{equation*}
Equivalently, either $Q^{\ast}>0$ and $\Lambda^{\ast}_{P,S}(Q^{\ast})=S\mu$, or $Q^{\ast}=0$ and $\Lambda^{\ast}_{P,S}(0)\le S\mu$. Then, we first establish strict decrease of the demand map above the residual floor.
\begin{lemma}
\label{lem:strict_decrease_residual}
Assume $D(\cdot)$ is strictly increasing.
Fix $(P,S)$ and define
\(
\mathcal{U}(P,S):=\{Q\ge 0:\Lambda^{\ast}_{P,S}(Q)>\Lambda_0(P)\}.
\)
Then $Q\mapsto\Lambda^{\ast}_{P,S}(Q)$ is strictly decreasing on $\mathcal{U}(P,S)$.
Consequently, for any constant $c>\Lambda_0(P)$, the equation $\Lambda^{\ast}_{P,S}(Q)=c$ has at most one solution.
\end{lemma}
\begin{proof}
Take $Q_2>Q_1$ with $Q_1,Q_2\in\mathcal{U}(P,S)$.
Then at least one delay-sensitive type ($\delta_k>0$) is active at $Q_1$; otherwise $\Lambda^{\ast}_{P,S}(Q_1)=\Lambda_0(P)$.
For such a type, strict monotonicity of $D(\cdot)$ gives
$P+\delta_k D(Q_2)>P+\delta_k D(Q_1)$, hence $\lambda_k^{\ast}(Q_2)<\lambda_k^{\ast}(Q_1)$.
All other type contributions are non-increasing, so
$\Lambda^{\ast}_{P,S}(Q_2)<\Lambda^{\ast}_{P,S}(Q_1)$.
\end{proof}
Following Lemma~\ref{lem:strict_decrease_residual}, we provide the result regarding the existence and uniqueness of the operating point as follows. 

\begin{proposition}[Existence and uniqueness]
\label{prop:unique_operating_point}
Assume $D(\cdot)$ is continuous and strictly increasing with $D(0)=0$ and $\lim_{Q\to\infty}D(Q)=\infty$.
If the guardrail condition \eqref{eq:guardrail} holds, then an operating point $Q^{\ast}$ exists and is unique.
Moreover, $Q^\ast=0$ if and only if $\Lambda^\ast_{P,S}(0)\le S\mu$; otherwise $Q^\ast>0$ and satisfies $\Lambda^\ast_{P,S}(Q^\ast)=S\mu$.
\end{proposition}
\begin{proof}
Let $f(Q):= \Lambda^{\ast}_{P,S}(Q)-S\mu$.
Under the stage-wise response in Section~\ref{sec:model}, $\Lambda^{\ast}_{P,S}(Q)$ is continuous and non-increasing in $Q$, hence so is $f$.
By \eqref{eq:guardrail} and Proposition~\ref{prop:residual_regime}, $ \lim_{Q\to\infty} f(Q)=\Lambda_0(P)-S\mu<0$.

If $f(0)\le 0$, then $Q^\ast=0$ is an operating point, and it is unique. 
Indeed, if $f(0)<0$, monotonicity gives $f(Q)<0$ for all $Q>0$, so no positive root exists; if $f(0)=0$, any other root $\widetilde Q>0$ would satisfy $\Lambda^{\ast}_{P,S}(\widetilde Q)=S\mu>\Lambda_0(P)$, hence $\widetilde Q\in\mathcal U(P,S)$, which is impossible by Lemma~\ref{lem:strict_decrease_residual}.

If $f(0)>0$, then continuity and $\lim_{Q\to\infty}f(Q)<0$ imply the existence of some $Q^\ast>0$ with $f(Q^\ast)=0$. Any such root must lie in $\mathcal U(P,S)$ because $S\mu>\Lambda_0(P)$, and Lemma~\ref{lem:strict_decrease_residual} then gives uniqueness.
\end{proof}

Now, to control transients, we next show that trajectories enter a bounded absorbing set in finite time.
\begin{lemma}
\label{lem:absorbing_region}
Assume \eqref{eq:guardrail} and define
\(
\Lambda_{\max}(P):= \sum_{k=1}^K\rho_k\Big[\frac{w_k}{P}-1\Big]_+.
\)
Then $\Lambda^{\ast}_{P,S}(Q)\le\Lambda_{\max}(P)$ for all $Q\ge 0$, and the interval
\(
\mathcal{Q}(P,S):= [0,Q^{\mathrm{drop}}(P,S)+\Delta\Lambda_{\max}(P)]
\) 
is forward invariant for \eqref{eq:fixed_backlog_recursion}, i.e., whenever $Q_t\in \mathcal{Q}(P,S)$, we also have $Q_{t+1}\in \mathcal{Q}(P,S)$. Moreover, every trajectory enters
$\mathcal{Q}(P,S)$ in finite time.
\end{lemma}
\begin{proof}
By Proposition~\ref{prop:neg_drift}, for $Q_t\ge Q^{\mathrm{drop}}(P,S)$ we have
$\Lambda^{\ast}_{P,S}(Q)-S\mu=-\varepsilon(P,S)<0$, so
$Q_{t+1}\le[Q_t-\Delta\varepsilon(P,S)]_+$.
Hence large backlogs decrease and reach $[0,Q^{\mathrm{drop}}(P,S)]$ in finite time.
Also, for $Q_t\le Q^{\mathrm{drop}}(P,S)$, $P+\delta_kD(Q)\ge P$ implies
$\lambda_k^{\ast}(Q)\le[\frac{w_k}{P}-1]_+$, so
$\Lambda^{\ast}_{P,S}(Q)\le\Lambda_{\max}(P)$.
Therefore, 
$Q_{t+1}\le Q^{\mathrm{drop}}(P,S)+\Delta\Lambda_{\max}(P)$, proving forward invariance and finite-time entry.
\end{proof}

We then impose a checkable step-size condition to guarantee order preservation of the induced map.

\begin{assumption}
\label{ass:stepsize}
There exists $L_\Lambda(P,S)$ such that $\Lambda^{\ast}_{P,S}$ is $L_\Lambda(P,S)$-Lipschitz on the forward-invariant region, and
\begin{equation}
\Delta\,L_\Lambda(P,S)\le 1.
\label{eq:stepsize_cond}
\end{equation}
\end{assumption}

Assumption~\ref{ass:stepsize} is checkable in practice: Lemma~\ref{lem:order_preserving} provides the explicit sufficient bound \eqref{eq:Lipschitz_bound}, which guarantees \eqref{eq:stepsize_cond}.
\begin{lemma}
\label{lem:order_preserving}
Under Assumption~\ref{ass:stepsize}, the map $F_{P,S}$ induced by \eqref{eq:fixed_backlog_recursion} is non-decreasing on the forward-invariant region. Moreover, the following bound is sufficient for \eqref{eq:stepsize_cond}:
\begin{equation}
L_\Lambda(P,S)\;\le\;
\frac{1}{S\mu+\epsilon}\cdot \frac{1}{P^2}\sum_{k:\,\delta_k>0}\rho_k w_k\delta_k.
\label{eq:Lipschitz_bound}
\end{equation}
\end{lemma}

\begin{proof}
Write $F_{P,S}(Q)=[G(Q)]_+$ with $G(Q) := Q+\Delta(\Lambda_{P,S}^{\ast}(Q)-S\mu)$. Since $x\mapsto [x]_+$ is non-decreasing, it suffices to show that $G$ is non-decreasing. For any $Q_2\ge Q_1$ in the forward-invariant region, Assumption~\ref{ass:stepsize} gives $\Lambda_{P,S}^{\ast}(Q_2)-\Lambda_{P,S}^{\ast}(Q_1)\ge -L_\Lambda(P,S)(Q_2-Q_1)$, hence $G(Q_2)-G(Q_1)\ge (1-\Delta L_\Lambda(P,S))(Q_2-Q_1)\ge 0$ by \eqref{eq:stepsize_cond}. Therefore $F_{P,S}$ is non-decreasing.

By \eqref{eq:best_response}, only types with $\delta_k>0$ contribute to the $Q$-dependence of $\Lambda_{P,S}^{\ast}$. For $\delta_k>0$, combining \eqref{eq:delay_proxy} with \eqref{eq:best_response} gives a $Q$-slope bounded in magnitude by $\frac{1}{S\mu+\epsilon}\cdot\frac{w_k\delta_k}{P^2}$, using $P+\delta_k D(Q)\ge P$. Summing over types with weights $\rho_k$ yields \eqref{eq:Lipschitz_bound}, which is therefore sufficient for \eqref{eq:stepsize_cond}.
\end{proof}

Then, in the following theorem combining Proposition~\ref{prop:unique_operating_point}, Lemma~\ref{lem:absorbing_region}, and Lemma~\ref{lem:order_preserving} yields global convergence.

\begin{theorem}[Global convergence to the operating point]
\label{thm:global_convergence}
Under the conditions of Proposition~\ref{prop:unique_operating_point} and Assumption~\ref{ass:stepsize},
the trajectory of \eqref{eq:fixed_backlog_recursion} satisfies $Q_t\to Q^{\ast}$ for any initial $Q_0\ge 0$.
\end{theorem}
\begin{proof}
By Lemma~\ref{lem:absorbing_region}, there exists $T<\infty$ such that
$Q_t\in\mathcal Q(P,S)$ for all $t\ge T$, where $\mathcal Q(P,S)$ is bounded and forward invariant.
Under the stage-wise response in Section~\ref{sec:model}, the map $F_{P,S}$ is continuous, and by Lemma~\ref{lem:order_preserving}, it is non-decreasing on $\mathcal Q(P,S)$.

If $Q_{T+1}\ge Q_T$, then
\(Q_{T+2}=F_{P,S}(Q_{T+1})\ge F_{P,S}(Q_T)=Q_{T+1}.\)
If $Q_{T+1}< Q_T$, then similarly
\(Q_{T+2}=F_{P,S}(Q_{T+1})\le F_{P,S}(Q_T)=Q_{T+1}.\)
Hence, in either case, by induction the tail sequence $\{Q_t\}_{t\ge T}$ is monotone.
Since $\mathcal Q(P,S)$ is bounded and forward invariant, $\{Q_t\}_{t\ge T}$ is bounded in $\mathcal Q(P,S)$, and therefore converges to some $\bar Q\in\mathcal Q(P,S)$.
Continuity of $F_{P,S}$ gives $\bar Q=F_{P,S}(\bar Q)$, so $\bar Q$ is an operating point.

By Proposition~\ref{prop:unique_operating_point}, the operating point is unique, and therefore $\bar Q=Q^\ast$ in either case.
Thus $Q_t\to Q^\ast$ for any initial $Q_0\ge 0$.
\end{proof}

\section{Guardrail-Aware Action Shield for Dynamic Control}
\label{sec:dynamic_control}
Theorem~\ref{thm:global_convergence} provides a local drainability certificate for fixed price--capacity pairs $(P,S)$. In the full  problem, however, the leader chooses a sequence of price--target actions $a_t=(P_t,S_t^{\mathrm{tar}})$ under provisioning frictions, and hence the fixed-pair analysis in Section~\ref{sec:fixed_leader_dynamics} does not by itself yield an end-to-end stability result for the closed-loop dynamic control problem posed in Section~\ref{sec:model}. 
We therefore use this certificate to filter dynamic price--target actions under provisioning frictions. 
This section develops an online action shield for dynamic execution and establishes its residual-regime negative-drift certificate.
\subsection{Safe-Price Map and Effective-Capacity Proxy}
To improve the robustness of the online shield, we enforce a safety margin $\zeta\in(0,1)$:
\(\Lambda_0(P)\ \le\ (1-\zeta)\,S\mu.\) For $S\in[0,S_{\max}]$, we define the safe-price map as
\begin{equation*}
\Psafe(S)\!:=\!
\inf\Big\{P\!\in\![P_{\min},P_{\max}] : \Lambda_0(P)\!\le\! (1-\zeta)S\mu\Big\},
\end{equation*}
with $\Psafe(S)=+\infty$ if the set is empty.
Since $\Lambda_0(P)$ is non-increasing in $P$, $\Psafe(S)$ is non-increasing in $S$.


To evaluate whether a proposed target can support the margin guardrail under provisioning frictions, we construct a one-step effective-capacity proxy directly from the dynamics in \eqref{eq:At_def}--\eqref{eq:S_dyn}. Given
a state $s=(Q,S,B,S^{\mathrm{tar,prev}})$ and a candidate target $S^{\mathrm{tar}}$, define
\begin{align}
A(s) :=& \min\{\omega B,\; S_{\max}-S\},
\nonumber\\
S^{\mathrm{eff}}(s,S^{\mathrm{tar}})
:=
&\Proj_{[0,S_{\max}]}\!\Big(S + A(s) - \nu [S-S^{\mathrm{tar}}]_+\Big).
\nonumber\\
S^{\mathrm{cert}}(s,S^{\mathrm{tar}})
:=&\min\{S,S^{\mathrm{eff}}(s,\;S^{\mathrm{tar}})\}\nonumber
\end{align}
Here $A(s)$ is the one-step activation term inherited from \eqref{eq:At_def}, while
$S^{\mathrm{eff}}(s,S^{\mathrm{tar}})$ is the resulting one-step active capacity induced by the activation and scale-down terms in \eqref{eq:S_dyn}. 
Taking the minimum ensures that the certified capacity $S^{\mathrm{cert}}(s,S^{\mathrm{tar}})$ is conservative under scale-down.

\subsection{Guardrail-Aware Action Shield}
We now define an optimizer-agnostic online action shield that filters each proposed action before execution \cite{alshiekh2018safe}. The shield follows a minimal-intervention principle: when the proposed target can be certified safe, it leaves the target unchanged and adjusts only the price, raising it just enough to satisfy the guardrail. When such certification is infeasible, the shield instead defaults to an emergency action as a best-effort recovery move. For a candidate action $a=(P,S^{\mathrm{tar}})\in\mathcal{A}$, write $S^{\mathrm{cert}}$ as shorthand for $S^{\mathrm{cert}}(s,S^{\mathrm{tar}})$. 
The resulting action shield is then defined by
\begin{equation*}
\SafeFilter(s,a)\! :=\!
\begin{cases}
\!\big(\!\max\{P, \Psafe(S^{\mathrm{cert}})\},S^{\mathrm{tar}}\big),\!\!\!
&\!\! \text{if } \Psafe(S^{\mathrm{cert}})\!\!<\!\!+\infty,\\[1mm]
a_{\mathrm{emg}}\!:= \!(P_{\max}, S_{\max}),
&\!\! \text{otherwise.}
\end{cases}
\end{equation*}

\begin{proposition}[Dynamic negative-drift certificate]
Consider a state $s_t$ and a proposed action $a_t$ such that
$P_{\mathrm{safe}}(S_t^{\mathrm{cert}})<+\infty$, where
$S_t^{\mathrm{cert}}
:=S^{\mathrm{cert}}(s_t,S_t^{\mathrm{tar}})$.
Let $\widetilde a_t=\Gamma(s_t,a_t)$ and denote its executed
price by $\widetilde P_t$.
If $Q_t\ge Q^{\mathrm{drop}}(\widetilde P_t,S_t)$, then
\(
Q_{t+1}\le
[Q_t-\Delta\zeta S_t^{\mathrm{cert}}\mu]^+.
\)
Thus, whenever $S_t^{\mathrm{cert}}>0$, the shielded action
induces strictly negative backlog drift in the residual regime.
\end{proposition}
\begin{proof}
In the residual regime,
$\Lambda^*(s_t,\widetilde a_t)=\Lambda_0(\widetilde P_t)$.
By the shield definition,
$\Lambda_0(\widetilde P_t)\le
(1-\zeta)S_t^{\mathrm{cert}}\mu$, while
$S_t^{\mathrm{cert}}\le S_t$. Hence
$\Lambda^*(s_t,\widetilde a_t)-S_t\mu
\le-\zeta S_t^{\mathrm{cert}}\mu$.
Substitute into the backlog recursion gives the result.
\end{proof}
Since the shield acts only at execution time, it can be placed on top of a wide range of planning or learning methods without changing their internal logic.  The shield itself requires $O(K_0)$ work per candidate action, where $K_0$ is the number of residual tenant types, and can be combined with scalable RL in larger state--action spaces.
\subsection{Filtered Control Problem and Truncation Bound}
Let the executed actions be $\tilde a_t=\SafeFilter(s_t,a_t)$ for proposed $a_t\in\mathcal{A}$, and let dynamics be $s_{t+1}=f(s_t,\tilde a_t)$. The filtered Bellman equation is
\begin{equation*}
V^{\ast}(s)\!=\!\max_{a\in\mathcal{A}}
\Big\{ r\big(s,\SafeFilter(s,a)\big) + \gamma V^{\ast}\big(f(s,\SafeFilter(s,a))\big)\Big\},\ \forall s\in\mathcal{S}_{grid}.
\end{equation*}
On the bounded discretized planning domain used for computation, the filtered Bellman operator is a contraction, and therefore admits a unique optimal value function \cite{bertsekas2012dynamic}. In Section~\ref{sec:experiments}, we compute the planning benchmark by applying grid-based value iteration (VI) on this guarded discretization.
For finite-horizon planning, let $V_t^{(H)}$ denote the value function of backward dynamic programming with horizon $H$ (DP-$H$) at stage $t$, with terminal condition $V_H^{(H)}\equiv 0$. The standard geometric truncation bound then yields
\begin{align}
&\big|V^{\ast}(s)-V_0^{(H)}(s)\big|
\le \frac{\gamma^H}{1-\gamma}\,R_{\max},
\nonumber \\
&R_{\max} := \!\!\sup_{(s,a)\in \mathcal{S}_{\mathrm{grid}}\times \mathcal{A}_{\mathrm{grid}}}\!\!
\big|r(s,\SafeFilter(s,a))\big|<\infty.
\nonumber 
\end{align}
\section{Experiments}
\label{sec:experiments}

Having constructed the guarded control interface in Section~\ref{sec:dynamic_control}, we now evaluate its numerical behavior.
The conducted experiments address two questions. First, in Section~\ref{sec:exp_planning}, we test whether the guarded planning interface provides a reliable approximation to the VI benchmark both on the guarded discretization used for planning and under off-grid execution on the continuous dynamics.
Second, in Section~\ref{sec:exp_rl}, we evaluate whether the proposed guardrail-aware action shield improves safety and learning stability in model-free RL for the pricing-scaling task under both steady and bursty demand.

\begin{table}[t]
\vspace{5pt}
\centering
\caption{Baseline configuration.}
\label{tab:baseline}
\small
\setlength{\tabcolsep}{4pt}
\renewcommand{\arraystretch}{1.10}
\begin{tabular}{@{}p{0.43\linewidth}p{0.49\linewidth}@{}}
\toprule
Parameter(s) & Value(s) \\ \midrule
\multicolumn{2}{@{}l}{\textit{\textbf{Time, service, and bounds}}} \\
$\Delta,\ \gamma$ & $1.0,\ 0.99$ \\
$\mu,\ \epsilon$ & $0.80,\ 0.10$ \\
$Q$,\ $S$,\ $B$ & $[0,50],\ [0,4],\ [0,3]$ \\
$P$,\ $S^{\mathrm{tar}}$ & $[1,6],\ [0,4]$ \\
\addlinespace[2pt]
\multicolumn{2}{@{}l}{\textit{\textbf{Provisioning and reward}}} \\
$\omega,\ \nu,\ \chi$ & $0.2,\ 0.2,\ 0.3$ \\
$c_{\mathrm{op}},\ c_B$ & $1,\ 0.50$ \\
$\eta_{\mathrm{tar}},\ \Phi_0$ & $0.4,\ 8$ \\
\addlinespace[2pt]
\multicolumn{2}{@{}l}{\textit{\textbf{Safety margin}}} \\
$\zeta$ & $0.05$ \\
\bottomrule
\end{tabular}
\end{table}

\subsection{Baseline Setup and Evaluation Protocol}
\label{sec:exp_common}

\noindent\textbf{Baseline setup.}
All experiments share the same baseline instance and default guarded discretization.
Table~\ref{tab:baseline} summarizes system parameters, and Table~\ref{tab:tenant_types} specifies the heterogeneous tenant mix, including delay-sensitive types and delay-insensitive types.
n practice, the demand parameters $(w_k,\delta_k,\rho_k)$ can be calibrated from historical submission rates under observed prices and congestion, while $\mu$ and $(\omega,\nu,\chi)$ can be estimated from service-throughput and scaling traces. The safety margin $\zeta$ can further absorb moderate calibration errors in the residual-demand estimate.

We use the mixed-resolution backlog grid
$Q=\{0,0.5,1,\dots,5,6,7,\dots,50\}$,
uniform grids for state variables $S_t$, $B_t$, and $S^{\mathrm{tar,prev}}_t$ with step size $0.5$,
and uniform action grids for $S^{\mathrm{tar}}$ (step $0.5$) and $P$ (step $0.2$).

For each grid pair $(s,a)$, we first form the executed action $\tilde a=\SafeFilter(s,a)$, evaluate the one-step dynamics and reward, and project the resulting next state to the nearest grid point. This yields a deterministic guarded transition table and reward table.

\noindent\textbf{Evaluation protocol.}
For planning comparisons, we compute an infinite-horizon benchmark $(V^\infty,\pi^\infty)$ by VI and compare it with finite-horizon backward DP-$H$ on the same guarded transition and reward tables. For off-grid evaluation, these on-grid policies are executed on the continuous dynamics via nearest-neighbor lookup without per-step state projection \cite{bertsekas1995neuro}. We use the following metrics.

On-grid approximation quality is measured by the relative value gap at an initial state $s_0$,
\begin{equation*}
\!\!\!\!\mathrm{Gap}_{\mathrm{rel}}(H;s_0)
\!:=\!
\frac{|V^\infty(s_0)-V_0^{(H)}(s_0)|}
{\max\{|V^\infty(s_0)|,\varepsilon\}},\; \varepsilon=10^{-12}.
\end{equation*}
Here $V^\infty(s_0)$ denotes the infinite-horizon benchmark value at the initial state $s_0$ computed by VI, and $V_0^{(H)}(s_0)$ denotes the stage-$0$ value returned by finite-horizon backward DP with horizon $H$ at the same initial state.

Off-grid execution quality is measured by the mean relative return gap
\begin{equation*}
\mathrm{Gap}_{\mathrm{rel}}^{\mathrm{off}}(H)
:=
\mathbb{E}\!\left[
\frac{|J^{\mathrm{VI}}-J^{\mathrm{DP}(H)}|}
{\max\{|J^{\mathrm{VI}}|,\varepsilon\}}
\right],
\quad \varepsilon=10^{-12},
\end{equation*}
where the expectation is over perturbed initial states. 
Here $J^{\mathrm{VI}}$ and $J^{\mathrm{DP}(H)}$ denote the realized discounted returns obtained by executing the VI and the DP-$H$ policy on the continuous off-grid dynamics from the same initial state.

For model-free RL experiments, we also report the number of unsafe executed steps and the cumulative crash count. An executed step is counted as unsafe if the realized action violates the guardrail condition $\Lambda_0(\widetilde P_t)>(1-\zeta)S_t^{\mathrm{cert}}\mu$, and a crash occurs when the backlog reaches the truncation boundary $Q_{\max}$.

\begin{table}[t]
\vspace{-1mm}
\centering
\caption{Tenant-type parameters.}
\label{tab:tenant_types}
\small
\setlength{\tabcolsep}{4.5pt}
\renewcommand{\arraystretch}{1.10}
\begin{tabular}{@{}ccccc@{}}
\toprule
Type $k$ & $w_k$ & SLO$_k$ & $\delta_k$ & $\rho_k$ \\ \midrule
1 & 24 & 4 & 4.5 & 0.10 \\
2 & 18 & 6 & 2.5 & 0.15 \\
3 & 16 & 10 & 1.5 & 0.15 \\
4 & 15 & 15 & 1.1 & 0.10 \\
5 & 12 & $\infty$ & 0 & 0.30 \\
6 & 10 & $\infty$ & 0 & 0.20 \\
\bottomrule
\end{tabular}
\vspace{-2mm}
\end{table}
\subsection{Planning Interface Validation}
\label{sec:exp_planning}
\begin{figure}[b]
\centering
\includegraphics[width=0.88\columnwidth]{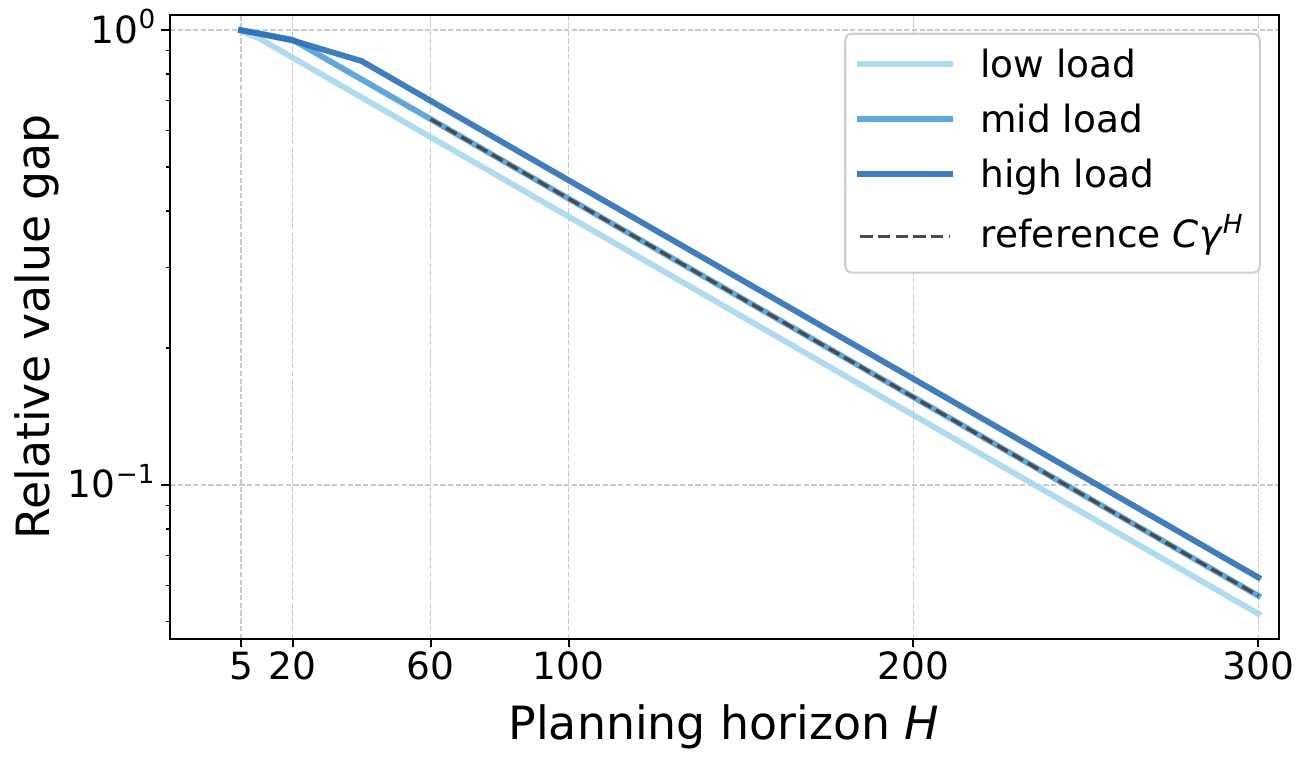}
\caption{Relative value gap versus planning horizon $H$; the dashed line is the geometric reference.}
\label{fig:exp1}
\end{figure}

We first test whether finite-horizon backward DP provides a reliable approximation to the VI benchmark on the guarded discretization.
To isolate backlog effects, we consider three representative initial loads, $q_{\mathrm{low}}=0.05Q_{\max}$, $q_{\mathrm{mid}}=0.5Q_{\max}$, and $q_{\mathrm{high}}=0.98Q_{\max}$, and set $s_0^\ell=(Q_0^\ell,S_0,B_0,S_0^{\mathrm{tar,prev}})=(q_\ell,0,0,0)$ for $\ell\in\{\mathrm{low},\mathrm{mid},\mathrm{high}\}$.
For each $s_0^{\ell}$, we compute the VI benchmark value $V^\infty(s_0^{\ell})$ and the finite-horizon value $V_0^{(H)}(s_0^\ell)$ on the same guarded table, and report the relative gap $\mathrm{Gap}_{\mathrm{rel}}(H;s_0^{\ell})$.
Fig.~\ref{fig:exp1} shows that $\mathrm{Gap}_{\mathrm{rel}}(H;s_0^{\ell})$ decreases monotonically with $H$ and exhibits an approximately geometric decay across all three initial-state backlogs.
The dashed reference curve is of the form $C\gamma^{H}$ anchored at $H_0=60$, illustrating that the observed convergence rate is consistent with the standard geometric truncation behavior in discounted DP. This supports the use of moderate finite horizons in the guarded planning interface.

\begin{figure}[t]
\centering
\hspace*{-5mm}
\includegraphics[width=0.88\columnwidth]
{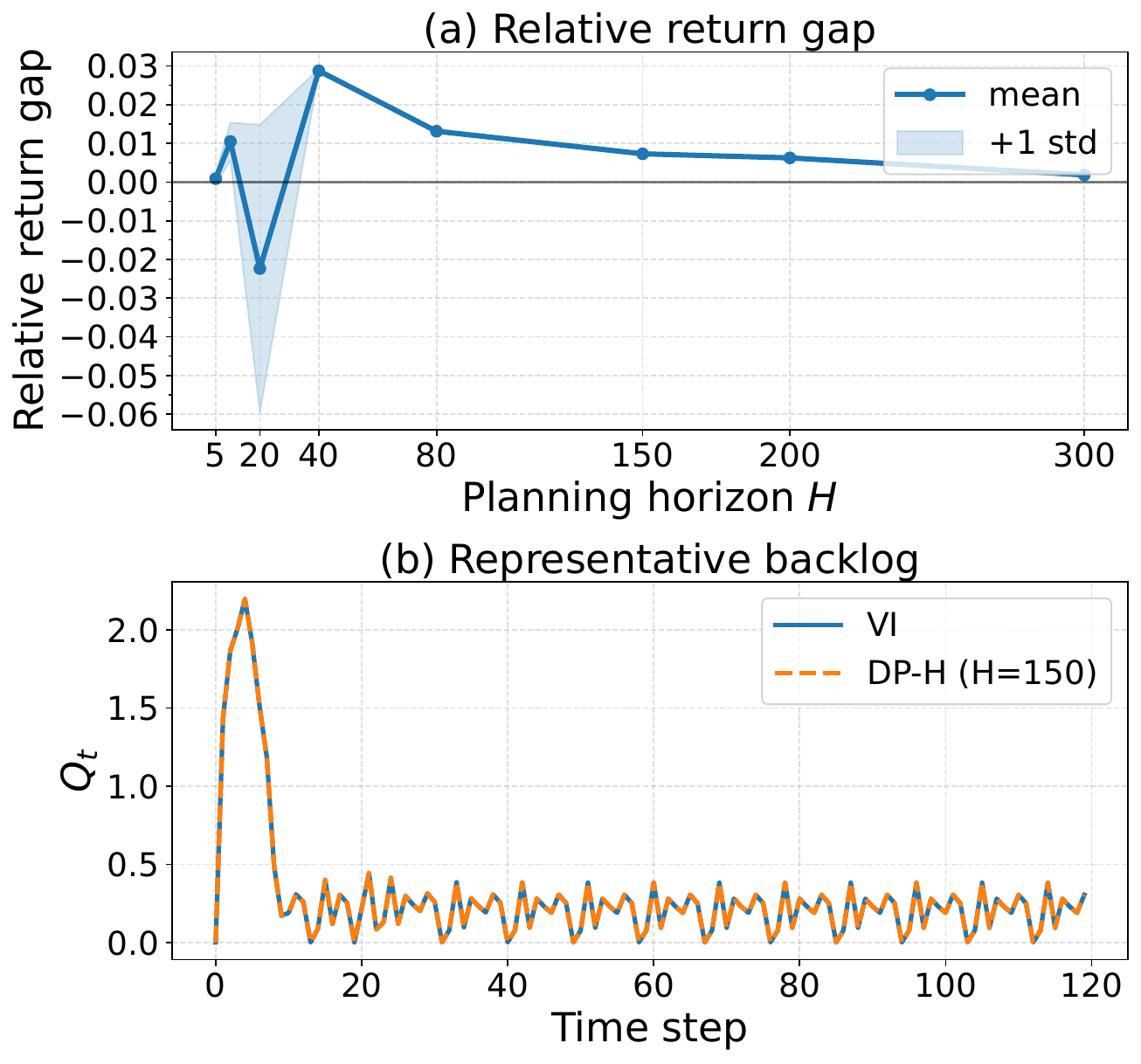}
\caption{
    (a) Relative off-grid return gap versus planning horizon $H$. 
    (b) Representative off-grid backlog trajectory for the VI and DP-$H$ policy.
    }
\label{fig:exp2_offgrid}

\end{figure}

We next assess how well on-grid planned policies transfer to the continuous off-grid dynamics, using the VI policy computed on the same guarded discretization as the reference for the resulting off-grid performance gap.
For each horizon $H$, we synthesize a DP-$H$ policy on the guarded table and execute it on the continuous system using nearest-neighbor lookup without per-step state projection; the VI policy $\pi^\infty$ is evaluated under the same off-grid protocol.
We report the mean relative return gap $\mathrm{Gap}_{\mathrm{rel}}^{\mathrm{off}}(H)$, where the expectation is over $n_{\mathrm{samp}}=30$ initial states with $Q_0$ uniformly perturbed within $\pm 5$ (clipped to $[0,Q_{\max}]$) and all other coordinates fixed.
For representative trajectory plots, we use the reference initial state $s_0=(0,0,0,0)$ over the first $120$ steps.
Fig.~\ref{fig:exp2_offgrid} shows that $\mathrm{Gap}_{\mathrm{rel}}^{\mathrm{off}}(H)$ is larger only at very short horizons and decays quickly, remaining small for moderate $H$, indicating that truncation bias dominates the short-$H$ regime.
Consistently, for a representative long horizon $H=150$, DP-$H$ and VI yield nearly identical off-grid backlog trajectories over the first $120$ steps.

\subsection{Guardrail Effects in Model-Free RL}
\label{sec:exp_rl}
\begin{figure}[t]
\centering
\includegraphics[width=\columnwidth]{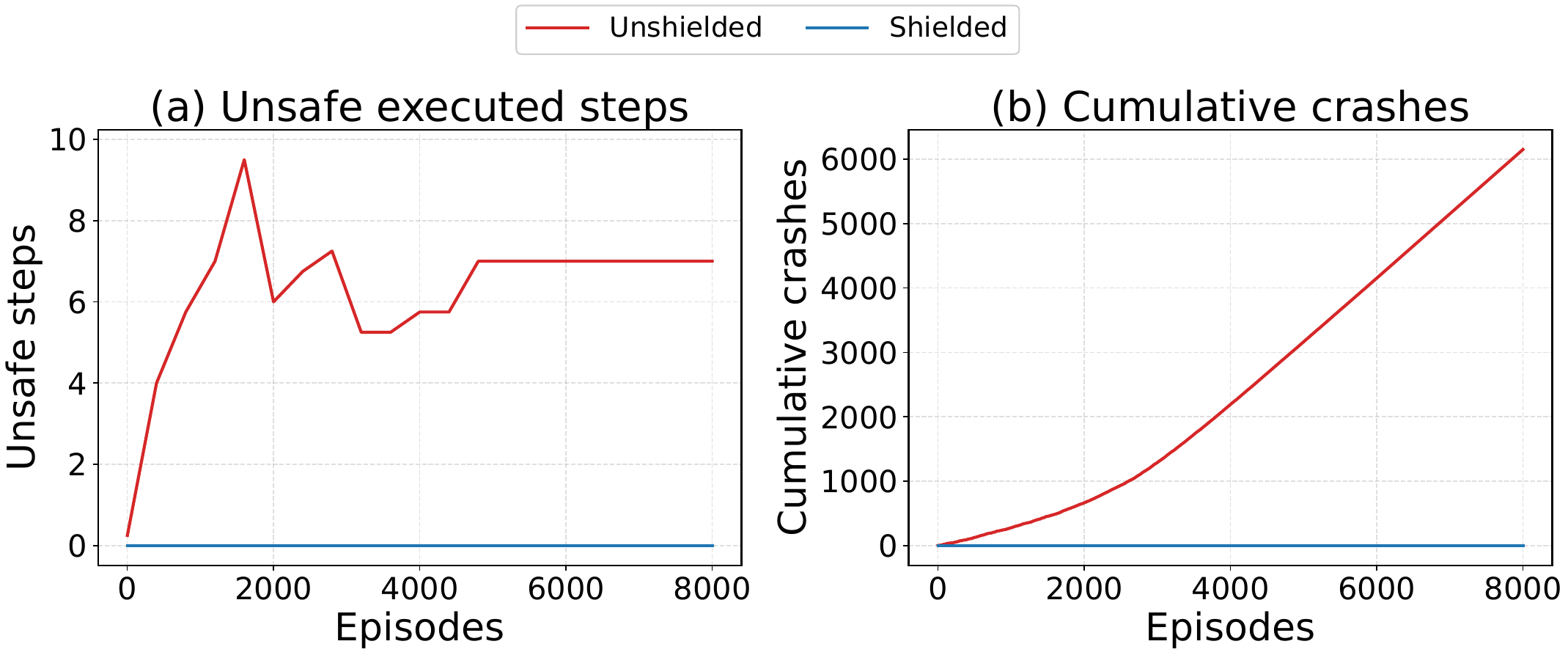}
\caption{Guardrail ablation in tabular Q-learning.}
\label{fig:exp3_learning}
\vspace{-4mm}
\end{figure}

We first compare unshielded and shielded tabular Q-learning to isolate the action shield as an execution-time safety layer.  Training uses $8000$ episodes of length $200$, each with step size $\alpha=0.15$ and an $\epsilon$-greedy schedule decaying from $1.0$ to $0.05$ over the first $80\%$ of training.
Every $400$ episodes, we run greedy evaluation rollouts of length $250$ and record the mean number of unsafe executed steps; we also track the cumulative number of training episodes that crash by hitting the truncation boundary. As shown in Fig.~\ref{fig:exp3_learning}, the guardrail-aware action shield drives unsafe executions to zero and substantially suppresses crash accumulation, indicating that it blocks unsafe exploratory actions that would trigger rapid backlog growth. Without the action shield, exploration repeatedly enters a residual-demand regime, leading to frequent crashes. Moreover, the shielded and unshielded policies achieve average evaluation returns of 339.04 and 30.64, respectively. These results corroborate the drainability analysis and highlight the action shield as an effective mechanism for safer exploration.

\begin{figure}[tb]
\centering
\includegraphics[width=0.99\columnwidth]{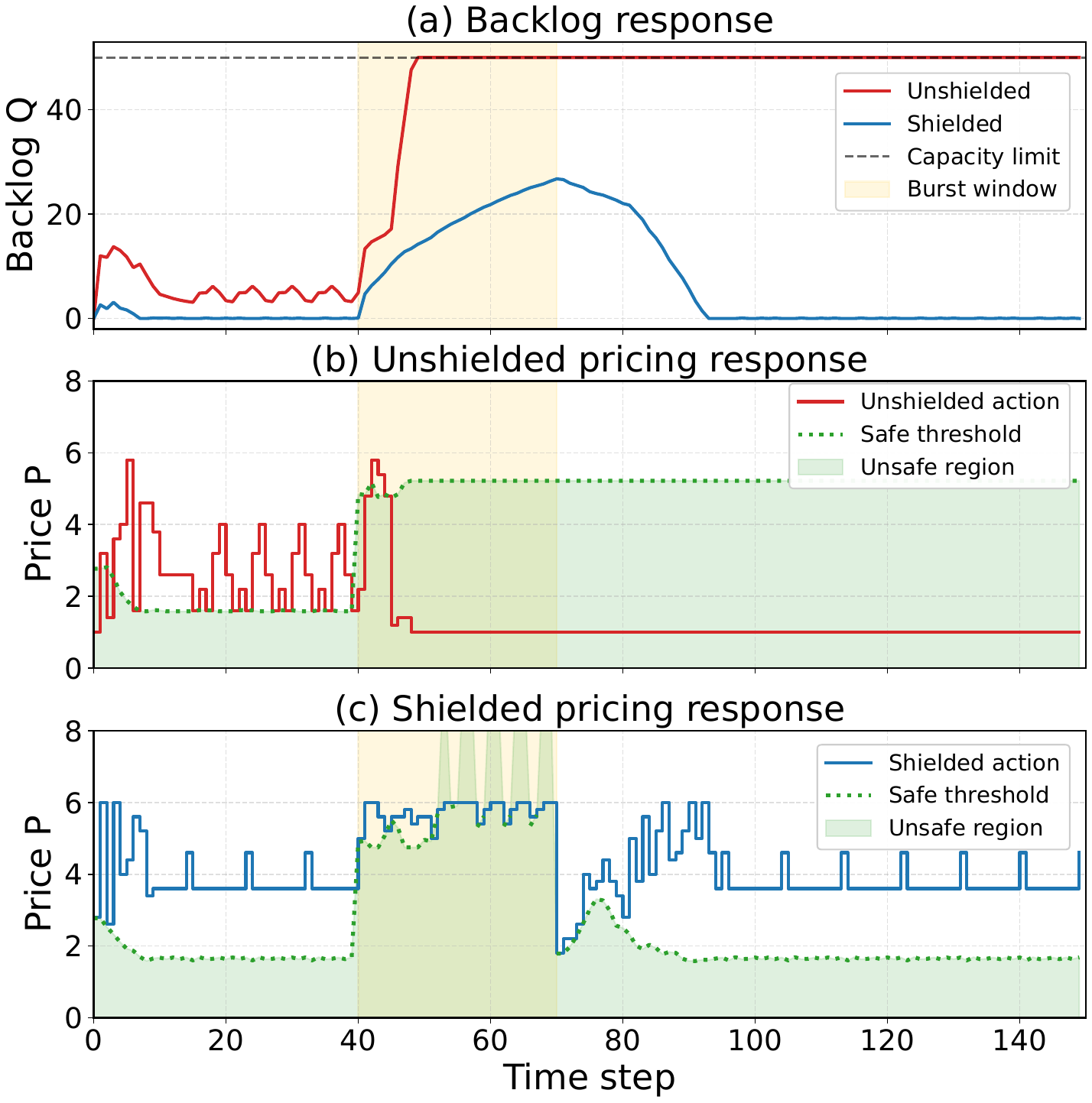}
\caption{Burst demand shift test. (a) Backlog response under shielded and unshielded execution. (b) Unshielded executed price versus its guardrail threshold. (c) Shielded executed price versus its guardrail threshold.}
\label{fig:exp4_burst}
\end{figure}

Finally, we test the shielded policy under a burst demand shift.
To this end, we stress-test the learned policy under an out-of-distribution demand shift from the steady-demand training setting \cite{fujimoto2024assessing}. Concretely, we scale the utility weights of all tenant types by a factor of $3$ during the time window $t\in[40,70)$, i.e., $w_k(t)=3w_k$ during the burst and $w_k(t)=w_k$ otherwise. We use a mild initial condition $s_0^{\mathrm{burst}}=(0,2,2,1)$ so that both policies remain well-behaved before the burst. We evaluate greedy policies learned under steady demand, with no exploration during evaluation: the unshielded policy executes $a_t=\arg\max_a \widehat Q(s_t,a)$, while the shielded policy executes $\tilde a_t=\SafeFilter(s_t,a_t)$ using a time-varying guardrail recomputed under the current $w_k(t)$.


Fig.~\ref{fig:exp4_burst} shows a clear separation. The unshielded policy frequently proposes prices below the guardrail threshold, enters an effectively undrainable regime, and drives the backlog toward the truncation boundary. In contrast, shielded execution keeps the realized actions guardrail-feasible throughout, raises the executed price to the elevated safe threshold during the burst, maintains bounded backlog dynamics, and returns to a stable operating level after the burst subsides. Thus, the burst test further shows that the action shield improves robustness to abrupt demand shifts, whereas the unshielded policy fails to recover safely.

Taken together, these experiments demonstrate that guardrail-aware action shield substantially improves safety and stability in model-free RL under both steady and bursty demands.

\section{Conclusion and future work}
\label{sec:conclusion}
In this work, we have studied the joint pricing-and-scaling problem in multi-tenant GPU cloud platforms with provisioning frictions and heterogeneous endogenous demand through a large-population Stackelberg game formulation. The Stackelberg formulation reveals a residual-demand mechanism generated by delay-insensitive tenants that can render the backlog structurally undrainable under bounded price and capacity. For the fixed price--capacity recursion, we have derived a computable drainability guardrail and established uniqueness of the operating point and global convergence under a checkable step-size condition. We have then translated the
guardrail into an online action shield with a residual-regime negative-drift certificate for dynamic control. Experiments have shown that the proposed guardrail-aware action shield improves the safety and robustness of model-free RL for this pricing and scaling task under demand shifts.

Future work will extend the framework to stochastic demand and service dynamics.


\bibliographystyle{IEEEtran}
\bibliography{references}
\end{document}